\documentclass{amsart}
\usepackage{amssymb}
\usepackage{mathrsfs}

\newtheorem{theorem}{Theorem}
\newtheorem{lemma}[theorem]{Lemma}
\newtheorem{proposition}[theorem]{Proposition}
\newtheorem{corollary}[theorem]{Corollary}

\theoremstyle{definition}
\newtheorem{definition}[theorem]{Definition}

\theoremstyle{remark}
\newtheorem{remark}[theorem]{Remark}

\DeclareMathOperator{\sinc}{sinc}

\newcommand{\R}{{\mathbb R}}
\newcommand{\Z}{{\mathbb Z}}
\newcommand{\C}{{\mathbb C}}

\begin{document}

\title[Eigenfunction expansions]{Eigenfunction expansions for the Schr\"odinger equation with an inverse-square potential}

\author{A.G.~Smirnov}
\address{I.~E.~Tamm Theory Department, P.~N.~Lebedev
Physical Institute, Leninsky prospect 53, Moscow 119991, Russia}
\email{smirnov@lpi.ru}
\dedicatory{Dedicated to Professor I.V.~Tyutin on the occasion of his 75th birthday}

\begin{abstract}
We consider the one-dimensional Schr\"odinger equation $-f''+q_\kappa f = Ef$ on the positive half-axis with the potential $q_\kappa(r)=(\kappa^2-1/4)r^{-2}$. For each complex number $\vartheta$, we construct a solution $u^\kappa_\vartheta(E)$ of this equation that is analytic in $\kappa$ in a complex neighborhood of the interval $(-1,1)$ and, in particular, at the ``singular'' point $\kappa = 0$. For $-1<\kappa<1$ and real $\vartheta$, the solutions $u^\kappa_\vartheta(E)$ determine a unitary eigenfunction expansion operator $U_{\kappa,\vartheta}\colon L_2(0,\infty)\to L_2(\R,\mathcal V_{\kappa,\vartheta})$, where $\mathcal V_{\kappa,\vartheta}$ is a positive measure on $\R$. We show that every self-adjoint realization of the formal differential expression $-\partial^2_r + q_\kappa(r)$ for the Hamiltonian is diagonalized by the operator $U_{\kappa,\vartheta}$ for some $\vartheta\in\R$. Using suitable singular Titchmarsh-Weyl $m$-functions, we explicitly find the measures $\mathcal V_{\kappa,\vartheta}$ and prove their continuity in $\kappa$ and $\vartheta$.
\end{abstract}

\maketitle

\section{Introduction}
\label{intro}

This paper is devoted to eigenfunction expansions connected with the one-dimensional Schr\"odinger equation
\begin{equation}\label{eq1}
-\partial^2_r f(r) + \frac{\kappa^2-1/4}{r^2} f(r) = Ef(r),\quad r>0,
\end{equation}
where $\kappa$ and $E$ are real parameters. It is easy to see that the function $f(r) = r^{1/2} J_\kappa(E^{1/2}r)$, where $J_\kappa$ is the Bessel function of the first kind of order $\kappa$, is a solution of~(\ref{eq1}) for every $E>0$ and $\kappa\in\R$ (this follows immediately from the fact that $J_\kappa$ satisfies the Bessel equation). These solutions can be used to expand square-integrable functions on the positive half-axis $\R_+ = (0,\infty)$. More precisely, given $\kappa>-1$ and a square-integrable complex function $\psi$ on $\R_+$ that vanishes for large $r$, we can define the function $\hat \psi$ on $\R_+$ by setting
\begin{equation}\label{eq2}
\hat\psi(E) = \frac{1}{\sqrt{2}}\int_0^\infty \sqrt{r}J_{\kappa}(\sqrt{E} r)\psi(r)\,dr,\quad E>0.
\end{equation}
The map $\psi\to\hat\psi$ up to a change of variables then coincides with the well-known Hankel transformation~\cite{Hankel} and induces a uniquely determined unitary operator in $L_2(\R_+)$. Since the development of a general theory of singular Sturm-Liouville problems by Weyl~\cite{Weyl}, this transformation has been used by many authors to illustrate various approaches to eigenfunction expansions for this kind of problem~\cite{Weyl1,Titchmarsh,Naimark,GesztesyZinchenko,Fulton,KST}.

If $\kappa\geq 1$, then transformation~(\ref{eq2}) is the unique eigenfunction expansion associated with~(\ref{eq1}) up to normalization of eigenfunctions. On the other hand, for $|\kappa|<1$, a one-parametric family of different expansions can be constructed using solutions of~(\ref{eq1}) (see~Chap.~4 in~\cite{Titchmarsh}). The reason for this ambiguity is that the formal differential expression for the Hamiltonian
\begin{equation}\label{eq3}
-\partial^2_r + \frac{\kappa^2-1/4}{r^2}
\end{equation}
does not uniquely determine the quantum-mechanical problem for $|\kappa|<1$ and admits various self-adjoint realizations in $L_2(\R_+)$ that yield different eigenfunction expansions. In~\cite{GTV2010}, all self-adjoint realizations of~(\ref{eq3}) were characterized using suitable asymptotic boundary conditions and the corresponding eigenfunction expansions were explicitly found.

In both~\cite{Titchmarsh} and~\cite{GTV2010}, the cases $0<|\kappa|<1$ and $\kappa=0$ were treated separately and eigenfunction expansions for $\kappa = 0$ could not be obtained from those for $0<|\kappa|<1$ by taking the limit $\kappa\to 0$. This situation is not quite satisfactory from the physical standpoint. In particular, self-adjoint operators associated with~(\ref{eq3}) can be used to construct self-adjoint realizations of Aharonov-Bohm Hamiltonian~\cite{Smirnov2015}, in which case zero and nonzero $\kappa$ correspond to integer and noninteger values of the dimensionless magnetic flux through the solenoid. Hence, the existence of a well-defined limit $\kappa\to 0$ is necessary to ensure the continuous transition between integer and noninteger values of the flux in the Aharonov-Bohm model. Here, we propose a parametrization of self-adjoint realizations of~(\ref{eq3}) and corresponding eigenfunction expansions that is continuous in $\kappa$ on the interval $(-1,1)$ (and, in particular, at $\kappa=0$).

We now formulate our main results. Let $\lambda$ denote the Lebesgue measure on $\R$ and $C_0^\infty(\R_+)$ be the space of all smooth functions on $\R_+$ with compact support. Given a $\lambda$-a.e.\footnote{Throughout the paper, a.e. means either ``almost every'' or ``almost everywhere.''} defined function $f$ on $\R_+$, we let $[f]$ denote the equivalence class of $f$ with respect to the Lebesgue measure on $\R_+$ (i.e., the restriction of the measure $\lambda$ to $\R_+$). For every $\kappa\in\R$, differential expression~(\ref{eq3}) naturally determines the operator $\check h_\kappa$ in $L_2(\R_+)$ whose domain $D_{\check h_\kappa}$ consists of all elements $[f]$ with $f\in C_0^\infty(\R_+)$:
\begin{equation}\label{checkhkappa}
\begin{split}
& D_{\check h_\kappa} = \left\{ [f] : f\in C_0^\infty(\R_+)\right\}, \\
& \check h_\kappa[f] = [-f''+q_\kappa f],\quad f\in C_0^\infty(\R_+).
\end{split}
\end{equation}
Here, $q_\kappa$ denotes the potential term in~(\ref{eq3}),
\begin{equation}\label{qkappa}
q_\kappa(r) = \frac{\kappa^2-1/4}{r^2},\quad r\in \R_+.
\end{equation}
The operator $\check h_\kappa$ is obviously symmetric and hence closable. The closure of $\check h_\kappa$ is denoted by $h_\kappa$,
\begin{equation}\label{hkappadef}
h_\kappa = \overline{\check h_\kappa}.
\end{equation}
The self-adjoint extensions of $h_\kappa$ (or, equivalently, of $\check h_\kappa$) can be naturally interpreted as self-adjoint realizations of formal expression~(\ref{eq3}) (cf.~Remark~\ref{r_realiz} below).

For any $z,\kappa\in\C$, we define the function $u^\kappa(z)$ on $\R_+$ by the relation\footnote{For brevity, we let $u^\kappa(z|r)$ denote the value of the function $u^\kappa(z)$ at a point $r$: $u^\kappa(z|r)=(u^\kappa(z))(r)$.}
\begin{equation}\label{ukappa}
u^\kappa(z|r) = r^{1/2+\kappa}\mathcal X_\kappa(r^2 z),\quad r\in\R_+,
\end{equation}
where the entire function $\mathcal X_\kappa$ is given by
\begin{equation}\label{Xkappa}
\mathcal X_\kappa(\zeta) = \frac{1}{2^\kappa}\sum_{n=0}^\infty \frac{(-1)^n\zeta^n}{\Gamma(\kappa+n+1)n!2^{2n}},\quad \zeta\in\C.
\end{equation}
The function $\mathcal X_\kappa$ is closely related to Bessel functions: for $z\neq 0$, we have
\begin{equation}\label{bessel}
\mathcal X_\kappa(\zeta) = \zeta^{-\kappa/2}J_\kappa(\zeta^{1/2}).
\end{equation}
Because $J_\kappa$ satisfies the Bessel equation, it follows that
\begin{equation}\label{eqf}
-\partial^2_r u^{\pm\kappa}(z|r) +  q_\kappa(r)u^{\pm\kappa}(z|r) = zu^{\pm\kappa}(z|r),\quad r\in\R_+,
\end{equation}
for every $\kappa\in\C$ and $z\neq 0$.\footnote{Here and hereafter, we assume that the function $q_\kappa$ on $\R_+$ is defined by~(\ref{qkappa}) for all $\kappa\in\C$.} By continuity, this also holds for $z=0$. In particular, $u^{\pm\kappa}(E)$ are solutions of spectral problem~(\ref{eq1}) for every $\kappa,E\in \R$.

Given a positive Borel measure $\sigma$ on $\R$ and a $\sigma$-measurable complex function $g$, we let $\mathcal T^\sigma_g$ denote the operator of multiplication by $g$ in $L_2(\R,\sigma)$.\footnote{More precisely, $\mathcal T^\sigma_g$ is the operator in $L_2(\R,\sigma)$ whose graph consists of all pairs $(\varphi_1,\varphi_2)$ such that $\varphi_1,\varphi_2\in L_2(\R,\sigma)$ and $\varphi_2(E) = g(E)\varphi_1(E)$ for $\sigma$-a.e. $E$.} If $g$ is real, then $\mathcal T^\sigma_g$ is self-adjoint. For $\kappa>-1$, we define the positive Radon measure\footnote{We recall that a Borel measure $\sigma$ on $\R$ is called a Radon measure on $\R$ if $\sigma(K)<\infty$ for every compact set $K\subset \R$.} $\mathcal V_\kappa$ on $\R$ by the relation
\begin{equation}\label{measVkappa}
d\mathcal V_\kappa(E) = \frac{1}{2}\Theta(E) E^{\kappa}\, dE,
\end{equation}
where $\Theta$ is the Heaviside function, i.e., $\Theta(E)=1$ for $E\geq 0$ and $\Theta(E)=0$ for $E<0$. Let $L_2^c(\R_+)$ denote the subspace of $L_2(\R_+)$ consisting of all its elements vanishing $\lambda$-a.e. outside some compact subset of $\R_+$.

It is well known (see, e.g., \cite{Naimark,GesztesyZinchenko,GTV2010}) that the operator $h_\kappa$ is self-adjoint and can be diagonalized by Hankel transformation~(\ref{eq2}) for $\kappa\geq 1$.
In terms of functions $u^\kappa(z)$, this result can be formulated as follows.

\begin{theorem}\label{leig1}
Let $\kappa>-1$ and the measure $\mathcal V_\kappa$ on $\R$ be defined by~$(\ref{measVkappa})$. Then
there is a unique unitary operator $U_\kappa\colon L_2(\R_+)\to L_2(\R,\mathcal V_\kappa)$ such that
\begin{equation}\nonumber
(U_\kappa\psi)(E) = \int_0^\infty u^{\kappa}(E|r)\psi(r)\,dr,\quad \psi\in L_2^c(\R_+),
\end{equation}
for $\mathcal V_\kappa$-a.e. $E$. The operator $U^{-1}_\kappa \mathcal T^{\mathcal V_\kappa}_\iota U_\kappa$, where $\iota$ is the identity function on $\R$ (i.e., $\iota(E) = E$ for all $E\in \R$), is a self-adjoint extension of $h_\kappa$ that coincides with $h_\kappa$ for $\kappa\geq 1$.
\end{theorem}

By~(\ref{ukappa}) and~(\ref{bessel}), we have $u^\kappa(E|r) = E^{-\kappa/2}r^{1/2}J_{\kappa}(E^{1/2}r)$, $r\in\R_+$, for every $E>0$. The operator $U_\kappa$ hence coincides with transformation~(\ref{eq2}) up to normalization of eigenfunctions. We note that $h_\kappa=h_{|\kappa|}$ for all $\kappa\in \R$ and $h_\kappa$ is therefore diagonalized by $U_{|\kappa|}$ for all real $\kappa$ such that $|\kappa|\geq 1$. If $0\leq \kappa<1$, then $U^{-1}_\kappa \mathcal T^{\mathcal V_\kappa}_\iota U_\kappa$ is the Friedrichs extension of $h_\kappa$ (see~\cite{EverittKalf}).

We now turn to parametrizing all self-adjoint extensions of $h_\kappa$ in the case $-1<\kappa<1$. Let
\[
\mathscr O = \{\kappa\in \C: \kappa \neq \pm 1,\pm 2,\ldots\}.
\]
For $\kappa\in\mathscr O$ and $\vartheta,z\in\C$, we define the function $u^\kappa_\vartheta(z)$ on $\R_+$ by setting
\begin{equation}\label{wkappa}
u^\kappa_\vartheta(z) = \frac{u^{\kappa}(z)\sin(\vartheta+\vartheta_\kappa)-u^{-\kappa}(z)\sin(\vartheta-\vartheta_\kappa)}{\sin\pi\kappa},\quad \kappa\in\mathscr O\setminus\{0\},
\end{equation}
and
\begin{multline}\label{w(z)}
u^0_\vartheta(z|r) = \lim_{\kappa\to 0} u^\kappa_\vartheta(z|r)=\\= u^0(z|r)\cos\vartheta+\frac{2}{\pi}\left[\left(\log\frac{r}{2} + \gamma\right)u^0(z|r) -\sqrt{r}\,\mathcal Y(r^2 z)\right]\sin\vartheta,\quad r\in\R_+,
\end{multline}
where
\begin{equation}\label{varthetakappa}
\vartheta_\kappa=\frac{\pi\kappa}{2},
\end{equation}
the entire function $\mathcal Y$ is given by
\begin{equation}\nonumber
\mathcal Y(\zeta) = \sum_{n=1}^\infty \frac{(-1)^nc_n}{(n!)^2 2^{2n}} \zeta^n,\quad c_n = \sum_{j=1}^n \frac{1}{j},
\end{equation}
and $\gamma = \lim_{n\to\infty} (c_n -\log n)=0,577\ldots$ is the Euler constant.\footnote{To compute the limit of $u^\kappa_\vartheta(z|r)$ as $\kappa\to 0$, we must apply L'H\^{o}pital's rule and use the equality $\Gamma'(1+n)/\Gamma(1+n)=c_n-\gamma$ (see formula~(9) in~Sec.~1.7.1 in~\cite{Bateman}).}

Given $\alpha\in \R$, we set $R_\alpha = \{z\in\C: z= re^{i\alpha} \mbox{ for some }r\geq 0\}$ and
\begin{equation}\label{Calpha}
\C_\alpha = \C\setminus R_\alpha.
\end{equation}
Hence, $\C_\alpha$ is the complex plane with a cut along the ray $R_\alpha$.

The next statement shows that, in spite of its piecewise definition, the quantity $u^\kappa_\vartheta(z|r)$ is actually analytic in all its arguments.

\begin{lemma}\label{l_analyt}
There is a unique analytic function $F$ in the domain $\mathscr O\times\C\times\C\times\C_\pi$ such that $F(\kappa,\vartheta,z,r) = u^\kappa_\vartheta(z|r)$ for every $\vartheta,z\in\C$, $\kappa\in\mathscr O$, and $r\in\R_+$.
\end{lemma}
The proof of Lemma~\ref{l_analyt} is given in Appendix~\ref{appA}.
\medskip

For every $\kappa\in \mathscr O$ and $\vartheta,z\in\C$, equality~(\ref{eqf}) also holds for $u^\kappa_\vartheta(z)$ in place of $u^{\pm\kappa}(z)$ (this obviously follows from~(\ref{wkappa}) for $\kappa \in \mathscr O\setminus\{0\}$. By Lemma~\ref{l_analyt}, we can take the limit $\kappa\to 0$ and conclude that the same holds for $\kappa=0$).\footnote{Alternatively, we can express $u^0_\vartheta(z|r)$ in terms of the Bessel functions $J_0$ and $Y_0$ by means of the equality $\pi Y_0(\zeta) = 2\left(\gamma+\log\frac{\zeta}{2}\right) J_0(\zeta) - 2\mathcal Y(\zeta^2)$ (see formula~(33) in~Sec.~7.2.4 in~\cite{Bateman}) and use the Bessel equation.}

Further, for every $\kappa\in (-1,1)$ and $\vartheta\in\R$, we define a positive Radon measure $\mathcal V_{\kappa,\vartheta}$ on $\R$ as follows. If $0<|\kappa|<1$, then we set
\begin{equation}\label{measVkappatheta}
\mathcal V_{\kappa,\vartheta} =  \left\{
\begin{matrix}
\tilde {\mathcal V}_{\kappa,\vartheta},& \vartheta \in [-|\vartheta_\kappa|,|\vartheta_\kappa|]+\pi\Z,\\
\frac{\pi\sin\pi\kappa|E_{\kappa,\vartheta}|}{2\kappa\sin(\vartheta+\vartheta_\kappa)\sin(\vartheta-\vartheta_\kappa)}\delta_{E_{\kappa,\vartheta}}+\tilde{\mathcal V}_{\kappa,\vartheta},& \vartheta\in (|\vartheta_\kappa|,\pi-|\vartheta_\kappa|)+\pi\Z,
\end{matrix}
\right.
\end{equation}
where $\vartheta_\kappa$ is defined by~(\ref{varthetakappa}),
the positive Radon measure $\tilde{\mathcal V}_{\kappa,\vartheta}$ on $\R$ is given by
\begin{multline}\label{tildeVkappatheta}
d\tilde{\mathcal V}_{\kappa,\vartheta}(E) = \\ =\frac{1}{2}\frac{\Theta(E) \sin^2\pi\kappa}{E^{-\kappa}\sin^2(\vartheta+\vartheta_\kappa) - 2\cos\pi\kappa\sin(\vartheta+\vartheta_\kappa)\sin(\vartheta-\vartheta_\kappa) + E^{\kappa}\sin^2(\vartheta-\vartheta_\kappa)}\,dE
\end{multline}
and $\delta_{E_{\kappa,\vartheta}}$ is the Dirac measure at the point
\begin{equation}\label{Ekappavartheta}
E_{\kappa,\vartheta} = -\left(\frac{\sin(\vartheta+\vartheta_\kappa)}{\sin(\vartheta-\vartheta_\kappa)}\right)^{1/\kappa}.
\end{equation}
For $\kappa=0$, the measure $\mathcal V_{\kappa,\vartheta}$ is defined by taking the limit $\kappa\to 0$ in formulas~(\ref{measVkappatheta})--(\ref{Ekappavartheta}). This yields
\begin{equation}\label{measV0theta}
\mathcal V_{0,\vartheta} =  \left\{
\begin{matrix}
\tilde {\mathcal V}_{0,\vartheta},& \vartheta \in \pi\Z,\\
\frac{\pi^2|E_{0,\vartheta}|}{2\sin^2\vartheta}\delta_{E_{0,\vartheta}}+\tilde{\mathcal V}_{0,\vartheta},& \vartheta\notin \pi\Z,
\end{matrix}
\right.
\end{equation}
where
\begin{equation}\label{E0vartheta}
E_{0,\vartheta} = -e^{\pi\cot \vartheta}
\end{equation}
and the positive Radon measure $\tilde{\mathcal V}_{0,\vartheta}$ on $\R$ is given by
\begin{equation}\label{tildeV0theta}
d\tilde{\mathcal V}_{0,\vartheta}(E) = \frac{1}{2}\frac{\Theta(E)}{ (\cos\vartheta -\pi^{-1}\log E\sin\vartheta)^2 +\sin^2\vartheta} dE.
\end{equation}

The next theorem describes self-adjoint extensions of $h_{\kappa}$ for $-1<\kappa<1$ in terms of their eigenfunction expansions.
\begin{theorem}\label{leig2}
Let $-1<\kappa<1$. For every $\vartheta\in \R$, there is a unique unitary operator $U_{\kappa,\vartheta}\colon L_2(\R_+)\to L_2(\R,\mathcal V_{\kappa,\vartheta})$ such that
\begin{equation}\nonumber
(U_{\kappa,\vartheta}\psi)(E) = \int_0^\infty u^\kappa_\vartheta(E|r)\psi(r)\,dr,\quad \psi\in L_2^c(\R_+),
\end{equation}
for $\mathcal V_{\kappa,\vartheta}$-a.e. $E$. The operator
\begin{equation}\nonumber
h_{\kappa,\vartheta} = U_{\kappa,\vartheta}^{-1} \mathcal T^{\mathcal V_{\kappa,\vartheta}}_\iota U_{\kappa,\vartheta},
\end{equation}
where $\iota$ is the identity function on $\R$, is a self-adjoint extension of $h_\kappa$. Conversely, every self-adjoint extension of $h_\kappa$ is equal to $h_{\kappa,\vartheta}$ for some $\vartheta\in\R$.
Given $\vartheta,\vartheta'\in\R$, we have $h_{\kappa,\vartheta} = h_{\kappa,\vartheta'}$ if and only if $\vartheta-\vartheta'\in \pi\Z$.
\end{theorem}

For $\vartheta = \vartheta_\kappa$, we have $\mathcal V_{\kappa,\vartheta} = \mathcal V_\kappa$ and $u^\kappa(z)=u^\kappa_\vartheta(z)$ for all $z\in\C$, and the operator $U_{\kappa,\vartheta}$ therefore coincides with the Hankel transformation $U_\kappa$.

The expansions described by Theorem~\ref{leig2} have the advantage that neither the eigenfunctions $u^\kappa_\vartheta(E)$ nor the spectral measures\footnote{In this paper, the term ``spectral measure'' always refers to a certain positive measure on $\R$ whose precise definition is given in Proposition~\ref{t_eig}. This usage differs from that adopted in~\cite{Smirnov2015}, where this term was applied to projection-valued measures in a Hilbert space.} $\mathcal V_{\kappa,\vartheta}$ have any discontinuities at $\kappa = 0$. This follows from Lemma~\ref{l_analyt} and the next theorem.

\begin{theorem}\label{l_borel}
Let $\varphi$ be a continuous function or a bounded Borel function on $\R$ with compact support. Then $(\kappa,\vartheta)\to \int \varphi(E)\,d\mathcal V_{\kappa,\vartheta}(E)$ is respectively a continuous function or a Borel function on $(-1,1)\times\R$ that is bounded on $[-\alpha,\alpha]\times \R$ for every $0\leq\alpha<1$.
\end{theorem}

Our main results are Theorems~\ref{leig2} and~\ref{l_borel}. We also give a new proof of Theorem~\ref{leig1} based on locally defined singular $m$-functions (see below).

To prove Theorems~\ref{leig1} and~\ref{leig2}, we use a recently developed variant of the Titch\-marsh-Weyl-Kodaira theory~\cite{GesztesyZinchenko,KST}. In those papers,
a generalization of the notion of the Titch\-marsh-Weyl $m$-function was proposed that is applicable not only to problems with a regular endpoint but also to a broad class of Schr\"odinger operators with two singular endpoints. Using such singular $m$-functions leads to a notable simplification in the treatment of eigenfunction expansions in comparison with the general theory~\cite{Kodaira,Naimark} based on matrix-valued measures (but we note that the results in~\cite{GesztesyZinchenko,KST} for eigenfunction expansions can be easily derived from Kodaira's general approach~\cite{Kodaira}; see Remark~\ref{reig} below).

The paper is organized as follows. In Sec.~\ref{s_one}, we give the general theory concerning self-adjoint extensions of one-dimensional Schr\"odinger operators and their eigenfunction expansions. The main statement in that section, Proposition~\ref{t_eig}, is similar to Theorem~3.4 in~\cite{KST}, but unlike the latter gives a local version of the formula for the spectral measures. This allows using different $m$-functions for different regions of the spectral parameter. In Sec.~\ref{s_eig}, we give a proof of Theorem~\ref{leig1} illustrating this local approach to finding spectral measures and establish Theorem~\ref{leig2}. Section~\ref{s4} is devoted to the proof of Theorem~\ref{l_borel}.

\section{One-dimensional Schr\"odinger operators}
\label{s_one}

In this section, we recall basic facts~\cite{Naimark,Teschl,Weidmann} concerning self-adjoint extensions of one-dimensional Schr\"odinger operators and briefly describe the approach to eigenfunction expansions developed in~\cite{GesztesyZinchenko,KST}. A distinctive feature of the subsequent exposition is that it uses the notion of a boundary space (see Definition~\ref{d_bound} below) that can be viewed as a formalization of the concept of a self-adjoint boundary condition. Using boundary spaces allows treating the limit point and limit circle cases on equal footing whenever possible, which makes the presentation of results clearer.

Let $-\infty \leq a < b \leq \infty$, $\lambda_{a,b}$ be the restriction to $(a,b)$ of the Lebesgue measure $\lambda$ on $\R$, and $\mathcal D$ be the space of all complex continuously differentiable functions on $(a,b)$ whose derivative is absolutely continuous on $(a,b)$ (i.e., absolutely continuous on every segment $[c,d]$ with $a<c\leq d<b$).
Let $q$ be a complex locally integrable function on $(a,b)$. Given $z\in\C$, we let $l_{q,z}$ denote the linear operator from $\mathcal D$ to the space of complex $\lambda_{a,b}$-equivalence classes such that
\begin{equation}\label{lqz}
(l_{q,z} f)(r) = -f''(r)+ q(r) f(r) - z f(r)
\end{equation}
for $\lambda$-a.e. $r\in (a,b)$ and set
\begin{equation}\nonumber
l_q = l_{q,0}.
\end{equation}
For every $f\in\mathcal D$ and $z\in\C$, we have $l_{q,z}f = l_q f - z[f]$, where $[f]=[f]_{\lambda_{a,b}}$ is the $\lambda_{a,b}$-equivalence class of $f$.
For every $c\in (a,b)$ and all complex numbers $z$, $\zeta_1$, and $\zeta_2$, there is a unique solution $f$ of the equation $l_{q,z}f=0$ such that $f(c)=\zeta_1$ and $f'(c)=\zeta_2$. This implies that solutions of $l_{q,z}f=0$ constitute a two-dimensional subspace of $\mathcal D$.
The Wronskian $W_r(f,g)$ at a point $r\in (a,b)$ of any functions $f,g\in \mathcal D$ is defined by the
relation
\begin{equation}\label{wronskian}
W_r(f,g) = f(r)g'(r) - f'(r)g(r).
\end{equation}
Clearly, $r\to W_r(f,g)$ is an absolutely continuous function on $(a,b)$. If $f$ and $g$ are such that $r\to W_r(f,g)$ is a constant function on $(a,b)$ (in particular, this is the case when $f$ and $g$ are solutions of $l_{q,z}f=l_{q,z}g=0$ for some $z\in\C$), then its value is denoted by $W(f,g)$.
It follows immediately from~(\ref{wronskian}) that the identities
\begin{align}\label{Plucker}
&W_r(f_1,f_2)W_r(f_3,f_4) + W_r(f_1,f_3)W_r(f_4,f_2) + W_r(f_2,f_3)W_r(f_1,f_4) = 0,\\
& W_r(f_1 f_2,f_3f_4) = f_1(r)f_3(r)W_r(f_2,f_4) + W_r(f_1,f_3)f_2(r)f_4(r)\label{wrprod}
\end{align}
hold for any $f_1,f_2,f_3,f_4\in \mathcal D$ and $r\in (a,b)$.

In the rest of this section, we assume that $q$ is real. Let
\begin{equation}\nonumber
\mathcal D_q = \{f\in \mathcal D : f \mbox{ and } l_q f \mbox{ are both square-integrable on } (a,b)\}.
\end{equation}
A $\lambda_{a,b}$-measurable complex function $f$ is said to be left or right square-integrable on $(a,b)$ if respectively $\int_a^c |f(r)|^2\,dr <\infty$ or $\int_c^b
|f(r)|^2\,dx <\infty$ for any $c\in (a,b)$. The subspace of $\mathcal D$ consisting of left or right square-integrable on $(a,b)$ functions $f$ such that $l_qf$ is also respectively left or right square-integrable on $(a,b)$ is denoted by $\mathcal D_q^l$ or $\mathcal D_q^r$. We obviously have $\mathcal D_q = \mathcal D_q^l\cap \mathcal D_q^r$.
It follows from~(\ref{lqz}) by integrating by parts that
\[
\int_c^d ((l_{q,z} f)(r)g(r) - f(r) (l_{q,z} g)(r))\,dr = W_d(f,g) - W_c(f,g)
\]
for every $f,g\in \mathcal D$, $z\in\C$, and $c,d\in (a,b)$. This implies the existence of limits $W_a(f,g) = \lim_{r\downarrow a} W_r(f,g)$ and $W_b(f,g) = \lim_{r\uparrow b} W_r(f,g)$ respectively for every $f,g\in \mathcal D_q^l$ and $f,g\in \mathcal D_q^r$. Moreover, it follows that
\begin{equation}\label{anti}
\langle l_q f, [g]\rangle - \langle [f], l_q g\rangle = W_{b}(\bar f,g) - W_{a}(\bar f,g)
\end{equation}
for any $f,g\in \mathcal D_q$, where $\langle\cdot,\cdot\rangle$ is the scalar product in $L_2(a,b)$.

For any linear subspace $Z$ of $\mathcal D_q$, let $L_q(Z)$ be the linear operator in $L_2(a,b)$ defined by the relations
\begin{equation}\label{L_q(Z)}
\begin{split}
& D_{L_q(Z)} = \{[f]: f\in Z\},\\
& L_q(Z) [f] = l_q f,\quad f\in Z.
\end{split}
\end{equation}
We define the minimal operator $L_q$ by setting
\begin{equation}\label{L_q}
L_q = L_q(\mathcal D_q^0),
\end{equation}
where
\begin{equation}\label{D^0_q}
\mathcal D_q^0 = \{f\in \mathcal D_q: W_a(f,g)=W_b(f,g)=0\mbox{ for any }g\in \mathcal D_q\}.
\end{equation}
By~(\ref{anti}), the operator $L_q(Z)$ is symmetric if and only if $W_{a}(\bar f,g) = W_{b}(\bar f,g)$ for any $f,g\in Z$. In particular, $L_q$ is a symmetric operator.
Moreover, $L_q$ is closed and densely defined, and its adjoint $L_q^*$ is given by
\begin{equation}\label{lq*}
L_q^* = L_q(\mathcal D_q)
\end{equation}
(see~Lemma~9.4 in~\cite{Teschl}).
If $T$ is a symmetric extension of $L_q$, then $L^*_q$ is an extension of $T^*$ and hence of $T$. In view of~(\ref{lq*}), this implies that $T$ is of the form $L_q(Z)$ for some subspace $Z$ of~$\mathcal D_q$.

\begin{remark}\label{r_realiz}
Self-adjoint operators of the form $L_q(Z)$ can be naturally viewed as self-adjoint realizations of the differential expression $-d^2/dr^2 +q$. If $L_q(Z)$ is self-adjoint, then equality~(\ref{lq*}) and the closedness of $L_q$ imply that $L_q(Z)$ is an extension of $L_q$ because $L_q(\mathcal D_q)$ is an extension of $L_q(Z)$. Therefore, the self-adjoint realizations of the expression $-d^2/dr^2 +q$ are precisely the self-adjoint extensions of the minimal operator $L_q$.
\end{remark}

\begin{definition}\label{d_bound}
We say that a linear subspace $X$ of $\mathcal D_q^l$ is a left boundary space if
\begin{enumerate}
\item[1.] if $W_{a}(\bar f,g)=0$ for any $f,g\in X$ and
\item[2.] if $g\in X$ whenever $g\in \mathcal D_q^l$ satisfies the equality $W_{a}(\bar f,g)=0$ for all $f\in X$.
\end{enumerate}
Replacing $\mathcal D_q^l$ with $\mathcal D_q^r$ and $a$ with $b$, we obtain the definition of a right boundary space.
\end{definition}

\begin{definition}\label{d_lpc}
If $W_{a}(f,g)=0$ for any $f,g\in \mathcal D_q^l$, then $q$ is said to be in the limit point case (l.p.c.) at $a$. Otherwise $q$ is said to be in the limit circle case (l.c.c.) at $a$. Similarly, $q$ is said to be in the l.p.c. at $b$ if $W_{b}(f,g)=0$ for any $f,g\in \mathcal D_q^r$ and to be in the l.c.c. at $b$ otherwise.
\end{definition}

Clearly, $q$ is in the l.p.c. at $a$ or $b$ if and only if $\mathcal D_q^l$ or $\mathcal D_q^r$ is the respective unique left or right boundary space. Given $f\in\mathcal D_q^l$ or $f\in\mathcal D_q^r$, we set
\begin{equation}\label{dqf}
\mathcal D_{q,f}^l = \{g\in \mathcal D_q^l : W_a(\bar f,g) = 0\},\quad \mathcal D_{q,f}^r = \{g\in \mathcal D_q^r : W_b(\bar f,g) = 0\}.
\end{equation}
For every $E\in\R$, we let $S_{q,E}^l$ and $S_{q,E}^r$ denote the respective sets of all nontrivial real elements $f$ of $\mathcal D_q^l$ and $\mathcal D_q^r$ such that $l_{q,E}f=0$.

The next proposition reformulates well-known results concerning self-adjoint extensions of $L_q$ (see, e.g., Sec.~9.2 in~\cite{Teschl}) in the language of boundary spaces.

\begin{proposition}\label{p_bound}
Let $q$ be a real locally integrable function on $(a,b)$. Then the following statements hold:
\begin{itemize}
\item[$1.$] Let $X$ and $Y$ respectively be left and right boundary spaces. Then the operator $L_q(X\cap Y)$ is a self-adjoint extension of $L_q$.
\item[$2.$] Let $L_q(X\cap Y) = L_q(\tilde X\cap\tilde Y)$ for some left boundary spaces $X$ and $\tilde X$ and right boundary spaces $Y$ and $\tilde Y$. Then we have $X=\tilde X$ and $Y=\tilde Y$.
\item[$3.$] Let $E\in\R$ and $f\in S_{q,E}^l$ or $f\in S_{q,E}^r$. Then $\mathcal D_{q,f}^l$ or $\mathcal D_{q,f}^r$ is respectively a left or right boundary space.
\item[$4.$] Let $z\in\C$. Then $q$ is in l.c.c. at $a$ or at $b$ if and only if every $f\in \mathcal D$ such that $l_{q,z}f=0$ is respectively left or right square-integrable on $(a,b)$.
\item[$5.$] If $q$ is in l.p.c. either at $a$ or at $b$, then every self-adjoint extension of $L_q$ is equal to $L_q(X\cap Y)$ for some left boundary space $X$ and right boundary space $Y$.
\item[$6.$] Let $q$ be in l.c.c. at $a$ or $b$ and $E\in \R$. Then every left or right boundary space is respectively equal to $\mathcal D_{q,f}^l$ or $\mathcal D_{q,f}^r$ for some $f\in S_{q,E}^l$ or $f\in S_{q,E}^r$.
\end{itemize}
\end{proposition}

The operators of the form $L_q(X\cap Y)$, where $X$ and $Y$ are left and right boundary spaces, are called self-adjoint extensions of $L_q$ with separated boundary conditions.

\begin{remark}
As mentioned above, boundary spaces can be thought of as self-adjoint boundary conditions. In this sense, the domain of $L_q(X\cap Y)$ consists of (the $\lambda_{a,b}$-equivalence classes of) all elements of $\mathcal D_q$ satisfying the self-adjoint boundary conditions $X$ and $Y$ on the respective left and right.
\end{remark}

\begin{remark}
Let $f$ and $g$ be linear independent solutions of $l_{q,z}f=l_{q,z}g=0$, where $\mathrm{Im}\,z\neq 0$. Suppose $f$ satisfies a self-adjoint boundary condition at $a$ (i.e., belongs to some left boundary space). Let $A$ denote the set of all $\zeta\in\C$ such that $g+\zeta f$ belongs to some right boundary space. Then $A$ is either a one-point set or a circle depending on whether $q$ is in the l.p.c. or l.c.c. at $b$. Moreover, $A$ is the limit of the circles $A_c$ obtained by replacing $b$ with a regular endpoint $c\in (a,b)$ in the definition of $A$. Such a limit procedure was originally used by Weyl~\cite{Weyl} to distinguish between the l.p.c. and l.c.c.
\end{remark}

If $q$ is in the l.p.c. at both $a$ and $b$, then statement~1 in Proposition~\ref{p_bound} implies that the operator $L_q(\mathcal D_q)$ is self-adjoint. In view of~(\ref{lq*}), it follows that $L_q$ is self-adjoint.

For every $f\in\mathcal D_q^l$, we set
\begin{equation}\label{Lqf}
L_q^f = L_q(\mathcal D_{q,f}^l\cap \mathcal D_q^r).
\end{equation}

\begin{lemma}\label{l_lqf}
Let $E\in\R$ and $q$ be in the l.c.c. at $a$ and in the l.p.c. at $b$. Then the self-adjoint extensions of $L_q$ are precisely the operators $L_q^f$, where $f\in S_{q,E}^l$. For $f,g\in S_{q,E}^l$, the equality $L_q^f = L_q^g$ holds if and only if $g = cf$ for some real $c\neq 0$.
\end{lemma}
\begin{proof}
The first statement follows immediately from statements~1, 3, 5, and~6 in Proposition~\ref{p_bound}. Let $f,g\in S_{q,E}^l$. If $g = cf$, then we have
\begin{equation}\label{eqD}
\mathcal D^l_{q,f} = \mathcal D^l_{q,g}
\end{equation}
by~(\ref{dqf}) and therefore $L_q^f = L_q^g$. Conversely, if $L_q^f = L_q^g$, then statements~2 and~3 in Proposition~\ref{p_bound} imply equality~(\ref{eqD}). Because $f\in \mathcal D_{q,f}^l$ by~(\ref{dqf}), we conclude that $f\in \mathcal D_{q,g}^l$ and hence $W_a(g,f)=0$. Because $l_{q,z}f=l_{q,z}g=0$, it follows that $W(g,f)=0$, whence $g = c f$.
\end{proof}

We now consider the eigenfunction expansions associated with $L_q$.

Let $O\subset \C$ be an open set. We say that a map $u\colon O\to \mathcal D$ is a $q$-solution in $O$ if $l_{q,z}u(z) = 0$ for every $z\in O$. A $q$-solution $u$ in $O$ is said to be analytic if the functions $z\to u(z|r)$ and $z\to \partial_r u(z|r)$ are analytic in $O$ for any $r\in (a,b)$. A $q$-solution $u$ in $O$ is said to be nonvanishing if $u(z)\neq 0$ for every $z\in O$ and is said to be real if $u(E)$ is real for every $E\in O\cap\R$.

\begin{definition}\label{d_triple}
A triple $(q,Y,u)$ is called an expansion triple if $q$ is a real locally integrable function on $(a,b)$, $Y$ is a right boundary space, and $u$ is a real nonvanishing analytic $q$-solution in $\C$ satisfying the following conditions:
\begin{itemize}
\item[1.] $u(z)\in \mathcal D_q^l$ for all $z\in\C$ and
\item[2.] there exists $E\in\R$ such that $W_a(u(E),u(z))=0$ for all $z\in\C$.
\end{itemize}
\end{definition}

\begin{lemma}\label{l_triple}
Let $\mathfrak t = (q,Y,u)$ be an expansion triple. Then there is a unique left boundary space $X^{\mathfrak t}$ such that $u(z)\in X^{\mathfrak t}$ for all $z\in \C$. For every $E\in\R$, we have $X^{\mathfrak t} = \mathcal D_{q,u(E)}^l$.
\end{lemma}
\begin{proof}
Let $E\in \R$ and $X$ be a left boundary space containing $u(E)$. By~(\ref{dqf}) and condition~(1) in Definition~\ref{d_bound}, we have $X\subset \mathcal D_{q,u(E)}^l$. On the other hand, if $g\in \mathcal D_{q,u(E)}^l$, then we have $W_a(\bar f,g)=0$ for every $f\in X$ because $\mathcal D_{q,u(E)}^l$ is a left boundary space by statement~3 in Proposition~\ref{p_bound}. In view of condition~(2) in Definition~\ref{d_bound}, we conclude that $g\in X$ and hence $X=\mathcal D_{q,u(E)}^l$. This implies that $X^{\mathfrak t}$ (if it exists) is unique and equal to $\mathcal D_{q,u(E)}^l$ for all $E\in\R$.
By~(\ref{dqf}) and Definition~\ref{d_triple}, there exists $E\in\R$ such that $u(z)\in \mathcal D_{q,u(E)}^l$ for all $z\in\C$. This proves the existence of $X^{\mathfrak t}$.
\end{proof}

Let $\mathfrak t = (q,Y,u)$ be an expansion triple, $\tilde u$ be a real analytic $q$-solution in $\C$ such that $W(u(z),\tilde u(z))\neq 0$ for every $z\in \C$, and $v$ be a nonvanishing analytic $q$-solution in $\C_+$\footnote{As usual, $\C_+$ denotes the open upper half-plane of the complex plane: $\C_+=\{z\in\C: \mathrm{Im}\,z>0\}$.} such that $v(z)\in Y$ for all $z\in\C_+$ (such $\tilde u$ and $v$ always exist; see Lemma~2.4 in~\cite{KST} and Lemma~9.8 in~\cite{Teschl}). Then $W(v(z),u(z))\neq 0$ for every $z\in \C_+$ because we would otherwise have $u(z)\in X^{\mathfrak t}\cap Y$ and hence the self-adjoint operator $L_q(X^{\mathfrak t}\cap Y)$ would have an eigenvalue in $\C_+$.
We define the analytic function $\mathcal M^{\mathfrak t}_{\tilde u}$ in $\C_+$ by the relation
\begin{equation}\label{mathcalM}
 \mathcal M^{\mathfrak t}_{\tilde u}(z) = \frac{1}{\pi}\frac{W(v(z),\tilde u(z))}{W(v(z),u(z))W(u(z),\tilde u(z))}
\end{equation}
(this definition is obviously independent of the choice of $v$). Following~\cite{KST}, we call such functions singular Titchmarsh-Weyl $m$-functions. Below, we see that it is sometimes useful to consider $q$-solutions $\tilde u$ that are defined on an open set $O\subset \C$ other than the entire complex plane. In that case, we assume that $\mathcal M^{\mathfrak t}_{\tilde u}$ is defined on $O\cap \C_+$.

Let $L_2^c(a,b)$ denote the subspace of $L_2(a,b)$ consisting of all its elements vanishing $\lambda$-a.e. outside some compact subset of $(a,b)$. The next proposition gives a way of constructing eigenfunction expansions for self-adjoint extensions of $L_q$ with separated boundary conditions.

\begin{proposition}\label{t_eig}
Let $\mathfrak t = (q,Y,u)$ be an expansion triple. Then the following statements hold:
\begin{itemize}
\item[$1.$] There exists a unique positive Radon measure $\sigma$ on $\R$ (called the spectral measure for $\mathfrak t$) such that
\[
\int \varphi(E)\, \mathrm{Im}\,\mathcal M^{\mathfrak t}_{\tilde u}(E+i\eta)\,dE\to \int \varphi(E)\,d\sigma(E)\quad (\eta\downarrow 0)
\]
for every continuous function $\varphi$ on $\R$ with compact support and every real analytic $q$-solution $\tilde u$ in $\C$ such that $W(u(z),\tilde u(z))\neq 0$ for every $z\in\C$.
\item[$2.$] Let $\sigma$ be the spectral measure for $\mathfrak t$. There is a unique unitary operator $U\colon L_2(a,b)\to L_2(\R,\sigma)$ (called the spectral transformation for $\mathfrak t$) such that
\begin{equation}\nonumber
(U\psi)(E) = \int_a^b u(E|r)\psi(r)\,dr,\quad \psi\in L_2^c(a,b),
\end{equation}
for $\sigma$-a.e. $E$.
\item[$3.$] Let $\sigma$ and $U$ be the spectral measure and transformation for $\mathfrak t$, and let the left boundary space $X^{\mathfrak t}$ be as in Lemma~$\ref{l_triple}$. Then we have
\begin{equation}\nonumber
L_q(X^{\mathfrak t}\cap Y) = U^{-1}\mathcal T^\sigma_\iota U,
\end{equation}
where $\iota$ is the identity function on $\R$.
\item[$4.$] Let $\sigma$ be the spectral measure for $\mathfrak t$, $O\subset\C$ be an open set, and $\tilde u$ be a real analytic $q$-solution in $O$ such that $W(u(z),\tilde u(z))\neq 0$ for every $z\in O$. Then we have
    \[
    \int_{\mathrm{supp}\,\varphi} \varphi(E)\, \mathrm{Im}\,\mathcal M^{\mathfrak t}_{\tilde u}(E+i\eta)\,dE\to \int_{O\cap\R} \varphi(E)\,d\sigma(E)\quad (\eta\downarrow 0)
    \]
    for every continuous function $\varphi$ on $O\cap\R$ with compact support ($\mathrm{supp}\,\varphi$ denotes the support of $\varphi$).
\end{itemize}
\end{proposition}
\begin{proof}
Statements~1--3 are a straightforward reformulation of the corresponding results in~\cite{KST} in the language of boundary spaces. Let $O$ and $\tilde u$ satisfy the conditions in statement~4 and $\theta$ be a real analytic $q$-solution in $\C$ such that $W(u(z),\theta(z))\neq 0$ for every $z\in\C$. Substituting $f_1 = u(z)$, $f_2 = v(z)$, $f_3 = \tilde u(z)$, and $f_4 = \theta(z)$ in~(\ref{Plucker}) and dividing the result by $\pi W(u(z),v(z))W(u(z),\theta(z))W(u(z),\tilde u(z))$ yields
\[
\mathcal M^{\mathfrak t}_{\tilde u}(z) =\mathcal M^{\mathfrak t}_{\theta}(z) + \frac{1}{\pi}\frac{W(\tilde u(z),\theta(z))}{W(u(z),\theta(z))W(u(z),\tilde u(z))}
\]
for any $z\in O\cap\C_+$. Statement~4 now follows from statement~1 because the last term in the right-hand side is analytic in $O$ and real on $O\cap\R$.
\end{proof}

\begin{corollary}\label{coreig}
Let $\sigma$ and $U$ be the spectral measure and transformation for an expansion triple $\mathfrak t=(q,Y,u)$. Then we have
\begin{equation}\label{v-1}
(U^{-1}\varphi)(r) = \int u(E|r) \varphi(E)\,d\sigma(E), \quad \varphi\in L_2^c(\R,\sigma),
\end{equation}
for $\lambda$-a.e. $r\in (a,b)$. If $\sigma(\{E\})\neq 0$ for some $E\in\R$, then $[u(E)]$ is an eigenfunction of $L_q(X^{\mathfrak t}\cap Y)$.
\end{corollary}
\begin{proof}
Given $\varphi\in L_2^c(\R,\sigma)$ and $r\in(a,b)$, let $\check \varphi(r)$ denote the right-hand side of~(\ref{v-1}). By statement~2 in Proposition~\ref{t_eig}, we have
\[
\langle \psi, U^{-1}\varphi\rangle = \langle U\psi, \varphi\rangle = \int d\sigma(E) \varphi(E)\int_a^b \overline{\psi(r)} u(E|r)\,dr = \int_a^b \overline{\psi(r)}\check \varphi(r)\,dr
\]
for any $\psi\in L^c_2(a,b)$, whence (\ref{v-1}) follows. In particular, we have $U^{-1}[\chi_{\{E\}}]_\sigma = \sigma(\{E\})[u(E)]$, where $\chi_{\{E\}}$ is the characteristic function of the one-point set $\{E\}$. By statement~3 in Proposition~\ref{t_eig}, this implies that $[u(E)]$ is an eigenfunction of $L_q(X^{\mathfrak t}\cap Y)$ if $\sigma(\{E\})\neq 0$.
\end{proof}

\begin{remark}\label{reig}
While the above proof of Proposition~\ref{t_eig} refers to~\cite{KST}, this result
can also be easily derived using Kodaira's general approach~\cite{Kodaira} based on matrix-valued measures. Indeed, if we set $s_1(z)=\tilde u(z)/W(u(z),\tilde u(z))$ and $s_2(z) = u(z)$ for $z\in\C$, then the only nonreal entry $M_{22}(z)$ of the characteristic matrix $M$ defined by formula~(1.13) in~\cite{Kodaira} is equal to $\pi\mathcal M^{\mathfrak t}_{\tilde u}(z)$ and statements~1--3 in Proposition~\ref{t_eig} hence essentially coincide with Theorem~1.3 in~\cite{Kodaira} in this case. The simple direct proof given in~\cite{KST} employs a single $m$-function and does not involve matrix-valued measures. It essentially relies on the technique developed in~\cite{GesztesyZinchenko}, where potentials in the l.p.c. at both endpoints were considered (a treatment in the same spirit for the l.c.c. at one of the endpoints can be found in~\cite{BennewitzEveritt}). A similar approach to finding spectral measures was also proposed in~\cite{GTV2010} in the context of the Schr\"odinger equation with the inverse-square potential.
\end{remark}

If $q$ is locally square-integrable on $(a,b)$, then formulas~(\ref{L_q}) and~(\ref{D^0_q}) imply that the space $C_0^\infty(a,b)$ of smooth functions on $(a,b)$ with compact support is contained in $\mathcal D_q^0$ and $L_q$ is an extension of $L_q(C_0^\infty(a,b))$. The proof of the next lemma is given in Appendix~\ref{app}.

\begin{lemma}\label{ll}
Let $q$ be a real locally square-integrable function on $(a,b)$. Then $L_q$ is the closure of $L_q(C_0^\infty(a,b))$.
\end{lemma}

\section{Eigenfunction expansions for inverse-square potential}
\label{s_eig}

We now assume that $a=0$ and $b=\infty$ and apply the above general theory to the potential $q_\kappa$ given by (\ref{qkappa}). It follows immediately from~(\ref{checkhkappa}) and~(\ref{L_q(Z)}) that $\check h_\kappa = L_{q_\kappa}(C^\infty_0(\R_+))$. In view of (\ref{hkappadef}) and Lemma~\ref{ll}, this implies that
\begin{equation}\label{hkappa}
h_\kappa = L_{q_\kappa},\quad \kappa\in\R.
\end{equation}

The equation $l_{q_\kappa} f = 0$ has linearly independent solutions $r^{1/2\pm\kappa}$ for $\kappa\neq 0$ and $r^{1/2}$ and $r^{1/2}\log r$ for $\kappa=0$. We conclude that by statement~4 in Proposition~\ref{p_bound},
\begin{itemize}
\item[1.] $q_\kappa$ is in the l.p.c. at both $0$ and $\infty$ for real $\kappa$ such that $|\kappa|\geq 1$ and
\item[2.]  $q_\kappa$ is in the l.p.c. at $\infty$ and in the l.c.c. at $0$ for $-1<\kappa<1$.
\end{itemize}
Hence, the operator $h_\kappa$ is self-adjoint for $|\kappa|\geq 1$ and has multiple self-adjoint extensions for $-1<\kappa<1$.

For any $\kappa\in\C$, let
the map $u^\kappa\colon \C\to \mathcal D$ be defined by~(\ref{ukappa}).
By~(\ref{eqf}), we have
\begin{equation}\label{lukappa}
l_{q_\kappa,z}u^{\pm\kappa}(z) = 0,\quad z,\kappa\in\C.
\end{equation}

In what follows, we systematically use notation~(\ref{Calpha}) for the complex plane with a cut along a ray. We let $\log$ denote the branch of the logarithm in $\C_{3\pi/2}$ satisfying the condition $\log 1 = 0$ and set $z^\rho=e^{\rho\log z}$ for all $z\in \C_{3\pi/2}$ and $\rho\in\C$.

For any $\kappa\in\C$, we define the map $v^\kappa\colon \C_{3\pi/2}\to \mathcal D$ by the relation
\begin{equation}\label{vkappa}
v^\kappa(z|r) = \frac{i\pi}{2} e^{i\pi\kappa/2} r^{1/2}H^{(1)}_{\kappa}(r z^{1/2}),\quad r\in\R_+,\, z\in\C_{3\pi/2},
\end{equation}
where $H^{(1)}_\kappa$ is the first Hankel function of order $\kappa$. Because $H^{(1)}_\kappa$ is a solution of the Bessel equation, we have
\begin{equation}\label{lvkappa}
l_{q_\kappa,z}v^{\kappa}(z) = 0
\end{equation}
for every $z\in\C_{3\pi/2}$ and $\kappa\in\C$.
It follows from the relation $H^{(1)}_{-\kappa} = e^{i\pi\kappa}H^{(1)}_{\kappa}$ (formula~(9) in~Sec.~7.2.1 in~\cite{Bateman}) that
\begin{equation}\label{-kappa}
v^{-\kappa}(z) = v^\kappa(z),\quad \kappa\in\C,\, z\in\C_{3\pi/2}.
\end{equation}
The well-known asymptotic form of $H^{(1)}_\kappa(\zeta)$ for $\zeta\to\infty$ (see formula~(1) in~Sec.~7.13.1 in~\cite{Bateman}) implies that
\begin{equation}\label{v_asym}
v^\kappa(z|r) \sim 2^{-1}\sqrt{\pi}(i+1)z^{-1/4} e^{iz^{1/2}r},\quad r\to\infty,
\end{equation}
for every $\kappa\in \C$ and $z\in\C_{3\pi/2}$ and hence
$v^\kappa(z)$ is right square-integrable for all $\kappa\in \C$ and $z\in\C_+$. Using the expression for the Wronskian of Bessel functions (formula~(29) in~Sec.~7.11 in~\cite{Bateman}),
\begin{equation}\label{wrbess}
W_z(J_\kappa,H^{(1)}_\kappa) = \frac{2i}{\pi z},
\end{equation}
and taking~(\ref{-kappa}) into account, from~(\ref{ukappa}), (\ref{bessel}), and~(\ref{vkappa}), we derive that
\begin{equation}\label{wrsol1}
W(v^\kappa(z),u^{\kappa}(z)) = z^{-\kappa/2}e^{i\pi\kappa/2},\quad W(v^\kappa(z),u^{-\kappa}(z)) = z^{\kappa/2}e^{-i\pi\kappa/2}
\end{equation}
for any $\kappa\in \C$ and $z\in\C_{3\pi/2}$.

\begin{lemma}\label{lemma_u}
Let $\kappa>-1$. Then $u^\kappa(z)$ is a nontrivial element of $\mathcal D_{q_\kappa}^l$ for every $z\in\C$, and we have
$W_0(u^\kappa(z),u^\kappa(z')) = 0$
for all $z,z'\in\C$.
\end{lemma}
\begin{proof}
Because $\kappa>-1$, it follows from~(\ref{ukappa}) that $u^\kappa(z)$ is left square-integrable for all $z\in\C$. In view of~(\ref{lukappa}), this implies that $u^\kappa(z)\in\mathcal D_{q_\kappa}^l$ for all $z\in\C$. By~(\ref{ukappa}), $u^\kappa(z)$ is nontrivial for $z\neq 0$ because otherwise $\mathcal X_\kappa$ would be identically zero. Because $u^\kappa(0|r)=2^{-\kappa}r^{1/2+\kappa}/\Gamma(\kappa+1)$ by~(\ref{ukappa}) and~(\ref{Xkappa}), we conclude that $u^\kappa(0)$ is nontrivial for $\kappa>-1$.
By~(\ref{ukappa}) and~(\ref{wrprod}), we have
\[
W_r(u^\kappa(z),u^\kappa(z')) = 2r^{2+2\kappa}(z'\mathcal X_\kappa(r^2z)\mathcal X'_\kappa(r^2z')-z\mathcal X'_\kappa(r^2z)\mathcal X_\kappa(r^2z'))
\]
and hence $W_0(u^\kappa(z),u^\kappa(z')) = 0$ for all $z,z'\in\C$.
\end{proof}

By~(\ref{lukappa}), $u^\kappa$ is a real analytic $q_\kappa$-solution in $\C$ for every $\kappa\in\R$. Because $q_\kappa$ is in the l.p.c. at $\infty$, $\mathcal D_{q_\kappa}^r$ is a right boundary space for all $\kappa\in\R$. Definition~\ref{d_triple} and Lemma~\ref{lemma_u} therefore imply that $\mathfrak t_\kappa = (q_\kappa,\mathcal D_{q_\kappa}^r,u^\kappa)$ is an expansion triple for every $\kappa>-1$. Let $\sigma_\kappa$ denote the spectral measure for $\mathfrak t_\kappa$.

\begin{lemma}\label{l_spec}
Let $\kappa>-1$. Then $\sigma_\kappa = \mathcal V_\kappa$, where $\mathcal V_\kappa$ is the measure on $\R$ defined by~$(\ref{measVkappa})$.
\end{lemma}
\begin{proof}
By~(\ref{vkappa}), (\ref{lvkappa}), and~(\ref{v_asym}), $v^\kappa$ is a nonvanishing analytic $q_\kappa$-solution in $\C_{3\pi/2}$ such that $v^\kappa(z)\in \mathcal D^r_{q_\kappa}$ for every $z\in\C_+$. Let $\tilde u_1$ be the restriction $v^\kappa|_O$ of $v^\kappa$ to the domain $O = \{z\in\C: \mathrm{Re}\,z < 0\}$. In view of~(\ref{vkappa}), we have $\tilde u_1(E) = r^{1/2} K_\kappa(r\sqrt{|E|})$ for $E<0$, where $K_\kappa$ is the modified Bessel function of the second kind of order $\kappa$ (formula~(15) in~Sec.~7.2.2 in~\cite{Bateman}). Hence, $\tilde u_1(E)$ is real for $E<0$. By~(\ref{wrsol1}), we have $W(\tilde u_1(z),u^\kappa(z))\neq 0$ for all $z\in O$. Substituting $\mathfrak t = \mathfrak t_\kappa$, $v = v^\kappa|_{\C_+}$, and $\tilde u = \tilde u_1$ in~(\ref{mathcalM}) yields $\mathcal M^{\mathfrak t_\kappa}_{\tilde u_1}(z) = 0$ for all $z\in O\cap\C_+$. By statement~4 in Proposition~\ref{t_eig}, we conclude that $\sigma_\kappa$ vanishes and hence coincides with $\mathcal V_\kappa$ on $(-\infty,0)$. Let the map $\tilde u_2\colon \C_\pi\to \mathcal D$ be given by $\tilde u_2(z|r) = r^{1/2} Y_\kappa(r z^{1/2})$, where $Y_\kappa$ is the Bessel function of the second kind of order $\kappa$. We have $l_{q_\kappa,z}\tilde u_2(z)=0$ for any $z\in \C_\pi$ because $Y_\kappa$ satisfies the Bessel equation, and $\tilde u_2$ is therefore an analytic $q_\kappa$-solution in $\C_\pi$. Because $Y_\kappa$ is real for positive real arguments, $\tilde u_2(E)$ is real for $E>0$. Because $H^{(1)}_\kappa = J_\kappa + i Y_\kappa$, it follows from~(\ref{wrbess}) that
\[
W_z(J_\kappa,Y_\kappa) = W_z(H^{(1)}_\kappa,Y_\kappa) = -iW_z(J_\kappa,H^{(1)}_\kappa) = \frac{2}{\pi z}.
\]
By~(\ref{ukappa}), (\ref{bessel}), and~(\ref{vkappa}), we obtain $W(u^\kappa(z),\tilde u_2(z)) = 2z^{-\kappa/2}/\pi \neq 0$ for $z\in \C_\pi$ and $W(v^\kappa(z),\tilde u_2(z)) = ie^{i\pi\kappa/2}$ for $z\in \C_+$. In view of~(\ref{wrsol1}), substituting $\mathfrak t = \mathfrak t_\kappa$, $v = v^\kappa|_{\C_+}$, and $\tilde u = \tilde u_2$ in~(\ref{mathcalM}) yields
\[
\mathcal M_{\tilde u_2}^{\mathfrak t_\kappa}(z) = \frac{iz^\kappa}{2},\quad z\in\C_+.
\]
Statement~4 in Proposition~\ref{t_eig} therefore ensures that $\sigma_\kappa$ coincides with $\mathcal V_\kappa$ on $(0,\infty)$. It remains to note that $\sigma_\kappa(\{0\})=0$ because otherwise $u^\kappa(0)$ would be square-integrable by Corollary~\ref{coreig}.
\end{proof}

\begin{proof}[Proof of Theorem~$\ref{leig1}$]
It follows from statement~2 in Proposition~\ref{t_eig} and Lemma~\ref{l_spec} that the operator $U_\kappa$ exists and is equal to the spectral transformation for $\mathfrak t_\kappa$. Statement~1 in Proposition~\ref{p_bound}, statement~3 in Proposition~\ref{t_eig}, formula~(\ref{hkappa}), and Lemma~\ref{l_spec} therefore imply that $U^{-1}_\kappa \mathcal T^{\mathcal V_\kappa}_\iota U_\kappa$ is a self-adjoint extension of $h_\kappa$. For $\kappa\geq 1$, $h_\kappa$ is self-adjoint and hence coincides with its self-adjoint extension $U^{-1}_\kappa \mathcal T^{\mathcal V_\kappa}_\iota U_\kappa$.
\end{proof}

\begin{remark}
As mentioned in Sec.~\ref{intro}, the operator $U_\kappa$ essentially coincides with the Hankel transformation. In~\cite{GesztesyZinchenko,KST}, where this transformation was treated similarly, the second solution $\tilde u$ used to calculate the spectral measure was required to be globally defined. This required distinguishing between integer and noninteger values of $\kappa$. Using a locally defined $\tilde u$ in the proof of Lemma~\ref{l_spec} allows treating all values of $\kappa$ uniformly.
\end{remark}

Given $\kappa\in\mathscr O$ and $\vartheta\in\C$, let the map $u^\kappa_\vartheta\colon \C\to \mathcal D$ be defined by~(\ref{wkappa}) and~(\ref{w(z)}). Because (\ref{eqf}) is satisfied for $u^\kappa_\vartheta(z)$ in place of $u^{\pm\kappa}(z)$ (see Sec.~\ref{intro}), we have
\begin{equation}\label{lukappavartheta}
l_{q_\kappa,z}u^\kappa_\vartheta(z)=0,\quad \kappa\in \mathscr O,\,\,\,\vartheta,z\in\C.
\end{equation}
By~(\ref{wkappa}) and~(\ref{wrsol1}), we have
\begin{equation}\label{Wuv}
W(v^\kappa(z),u^\kappa_\vartheta(z)) = \frac{z^{-\kappa/2}e^{i\pi\kappa/2}}{\sin\pi\kappa}(\sin(\vartheta+\vartheta_\kappa)-e^{-i\pi\kappa}z^\kappa \sin(\vartheta-\vartheta_\kappa))
\end{equation}
for all $\kappa\in \mathscr O\setminus\{0\}$, $\vartheta\in\C$, and $z\in\C_{3\pi/2}$. By Lemma~\ref{l_analyt}, $\kappa\to W(v^\kappa(z),u^\kappa_\vartheta(z))$ is an analytic function in $\mathscr O$ for fixed $\vartheta$ and $z$ and we can therefore find $W(v^0(z),u^0_\vartheta(z))$ by taking the limit $\kappa\to 0$ in~(\ref{Wuv}). As a result, we obtain
\begin{equation}\label{Wuv0}
W(v^0(z),u^0_\vartheta(z)) = \cos\vartheta + (i-\pi^{-1}\log z)\sin\vartheta,\quad \vartheta\in\C,\, z\in\C_{3\pi/2}.
\end{equation}
For every $\kappa\in\mathscr O$ and $z\in\C$, we set
\begin{equation}\label{wkappa(z)}
w^\kappa(z) = u^\kappa_{\pi/2+\vartheta_\kappa}(z),
\end{equation}
where $\vartheta_\kappa$ is given by~(\ref{varthetakappa}). It follows from~(\ref{wkappa}) and~(\ref{w(z)}) that
\begin{align}
&w^\kappa(z) = \frac{u^\kappa(z)\cos\pi\kappa - u^{-\kappa}(z)}{\sin\pi\kappa},\quad \kappa\in\mathscr O\setminus\{0\},\label{wk}\\
&w^0(z|r) = \frac{2}{\pi}\left[\left(\log\frac{r}{2} + \gamma\right)u^0(z|r) -\sqrt{r}\,\mathcal Y(r^2z)\right],\quad r\in\R_+,\label{w0}
\end{align}
for every $z\in\C$ and
\begin{equation}\label{ukappavartheta}
u^\kappa_\vartheta(z) =  u^\kappa(z)\cos(\vartheta-\vartheta_\kappa) + w^\kappa(z)\sin(\vartheta-\vartheta_\kappa)
\end{equation}
for all $\kappa\in\mathscr O$ and $\vartheta,z\in\C$.
By~(\ref{lukappavartheta}) and~(\ref{wkappa(z)}), we have
\begin{equation}\label{lu0}
l_{q_\kappa,z}w^\kappa(z)=0,\quad \kappa\in\mathscr O,\,\,z\in\C.
\end{equation}

\begin{lemma}\label{leigaux}
Let $-1<\kappa<1$. Then $u^\kappa(z),w^\kappa(z)\in \mathcal D^l_{q_\kappa}$ for every $z\in\C$, and
\[
W_0(w^\kappa(z),w^\kappa(z'))=0,\quad W_0(u^\kappa(z),w^\kappa(z')) = \frac{2}{\pi}
\]
for every $z,z'\in\C$.
\end{lemma}
\begin{proof}
Because $q_\kappa$ is in the l.c.c. at $0$ for $-1<\kappa<1$, statement~4 in Proposition~\ref{p_bound} and equalities~(\ref{lukappa}) and~(\ref{lu0}) imply that $u^\kappa(z),w^\kappa(z)\in \mathcal D^l_{q_\kappa}$ for every $z\in\C$.
Given $z\in\C$ and $-1<\kappa<1$, we define a smooth function $a^\kappa_z$ on $\R$ by setting $a^\kappa_z(r) = \mathcal X_\kappa(r^2z)$, where $\mathcal X_\kappa$ is given by~(\ref{Xkappa}). For $r\in\R_+$, we have $u^\kappa(z|r) = r^{1/2+\kappa}a^\kappa_z(r)$. In view of~(\ref{wrprod}), it follows that
\begin{equation}
W_r(u^\kappa(z),u^{-\kappa}(z')) = rW_r(a^\kappa_z,a^{-\kappa}_{z'}) - 2\kappa a^\kappa_z(r)a^{-\kappa}_{z'}(r)\nonumber
\end{equation}
for every $r\in\R_+$ and $z,z'\in\C$. Because $a^\kappa_z(0) = 2^{-\kappa}/\Gamma(1+\kappa)$ for any $z\in\C$, we obtain
$W_0(u^\kappa(z),u^{-\kappa}(z')) = -2\sin\pi\kappa/\pi$.
The statement of the lemma for $0<|\kappa|<1$ now follows from~(\ref{wk}) and Lemma~\ref{lemma_u}.
Given $z\in\C$, we define the smooth function $b_z$ on $\R$ by setting
\[
b_z(r) = (\gamma-\log 2)\mathcal X_{0}(r^2 z)-\mathcal Y(r^2z).
\]
By~(\ref{w0}), we have
\[
\pi w^0(z|r)/2 = r^{1/2}\log r\, a^0_z(r)+r^{1/2}b_z(r)
\]
for every $r\in\R_+$. In view of~(\ref{wrprod}), it follows that
\begin{align}
&\frac{\pi}{2}W_r(u^0(z),w^0(z')) =  r W_r(a^0_z,b_{z'})+r\log r W_r(a^0_z,a^0_{z'})+a^0_z(r)a^0_{z'}(r),\nonumber\\
&\frac{\pi^2}{4}W_r(w^0(z),w^0(z')) =  r\log^2rW_r(a^0_z,a^0_{z'})+r\log r(W_r(a^0_z,b_{z'})+W_r(b_z,a^0_{z'}))+\nonumber\\
&+ rW_r(b_z,b_{z'}) +b_z(r)a^0_{z'}(r)- a^0_z(r)b_{z'}(r)\nonumber
\end{align}
for every $r\in\R_+$ and $z,z'\in\C$. Because $a^0_z(0) = 1$  and $b_z(0) = \gamma-\log 2$ for any $z\in\C$, these equalities imply the required statement for $\kappa=0$.
\end{proof}

In view of~(\ref{lukappa}) and~(\ref{lu0}), Lemma~\ref{leigaux} implies that
\begin{equation}\label{wruw}
W(u^\kappa(z),w^\kappa(z)) = \frac{2}{\pi}
\end{equation}
for every $z\in\C$ and $-1<\kappa<1$ (and hence for all $z\in\C$ and $\kappa\in\mathscr O$), and it follows from~(\ref{ukappavartheta}) that
\begin{equation}\label{wrutildeu}
W(u^\kappa_\vartheta(z),u^\kappa_{\vartheta-\pi/2}(z)) = -\frac{2}{\pi},\quad \vartheta,z\in\C,\,\,\kappa\in\mathscr O.
\end{equation}

Let $-1<\kappa<1$ and $\vartheta\in \R$. By~(\ref{lukappavartheta}) and~(\ref{wrutildeu}), $u^\kappa_\vartheta$ is a real nonvanishing analytic $q_\kappa$-solution in $\C$. In view of~(\ref{ukappavartheta}), Lemma~\ref{lemma_u}, and Lemma~\ref{leigaux}, we have $u^\kappa_\vartheta(z)\in \mathcal D_{q_\kappa}^l$ for all $z\in \C$ and $W_0(u^\kappa_\vartheta(z),u^\kappa_\vartheta(z'))=0$ for all $z,z'\in\C$. Because $\mathcal D_{q_\kappa}^r$ is a right boundary space, it follows from Definition~\ref{d_triple} that $\mathfrak t_{\kappa,\vartheta} = (q_\kappa,\mathcal D_{q_\kappa}^r,u^\kappa_\vartheta)$ is an expansion triple. Let $\sigma_{\kappa,\vartheta}$ denote the spectral measure for $\mathfrak t_{\kappa,\vartheta}$.

\begin{lemma}\label{l_spec1}
Let $-1<\kappa<1$ and $\vartheta\in\R$. Then $\sigma_{\kappa,\vartheta} = \mathcal V_{\kappa,\vartheta}$, where $\mathcal V_{\kappa,\vartheta}$ is the measure on $\R$ defined by formulas~$(\ref{measVkappatheta})$--$(\ref{tildeV0theta})$.
\end{lemma}
\begin{proof}
By~(\ref{vkappa}), (\ref{lvkappa}), and~(\ref{v_asym}), $v^\kappa$ is a nonvanishing analytic $q_\kappa$-solution in $\C_{3\pi/2}$ such that $v^\kappa(z)\in \mathcal D^r_{q_\kappa}$ for every $z\in\C_+$.
Let the meromorphic function $\mathscr M_{\kappa,\vartheta}$ in $\C_{3\pi/2}$ be defined by the relation
\begin{equation}\label{Mkappavartheta}
\mathscr M_{\kappa,\vartheta}(z) = -\frac{1}{2}\frac{W(v^\kappa(z),u^\kappa_{\vartheta-\pi/2}(z))}{W(v^\kappa(z),u^\kappa_{\vartheta}(z))},\quad z\in \C_{3\pi/2}.
\end{equation}
Substituting $\mathfrak t=\mathfrak t_{\kappa,\vartheta}$ and $v = v^\kappa|_{\C_+}$ in~(\ref{mathcalM}) and taking~(\ref{wrutildeu}) into account, we conclude that $\mathscr M_{\kappa,\vartheta}$ coincides on $\C_+$ with the singular Titchmarsh-Weyl $m$-function $\mathcal M^{\mathfrak t_{\kappa,\vartheta}}_{\tilde u}$ for $\tilde u=u^\kappa_{\vartheta-\pi/2}$. By statement~1 in Proposition~\ref{t_eig}, we have
\begin{equation}\label{sigmakappavartheta}
\int \varphi(E)\,d\sigma_{\kappa,\vartheta}(E) = \lim_{\eta\downarrow 0} \int \varphi(E)\,\mathrm{Im}\mathscr M_{\kappa,\vartheta}(E+i\eta)\,dE
\end{equation}
for any continuous function $\varphi$ on $\R$ with compact support.
We note that $\sigma_{\kappa,\vartheta}(\{0\})=0$ because otherwise $u^\kappa_\vartheta(0)$ would be square-integrable by Corollary~\ref{coreig}. It therefore suffices to show that $\sigma_{\kappa,\vartheta}$ and $\mathcal V_{\kappa,\vartheta}$ coincide on the intervals $(-\infty,0)$ and $(0,\infty)$. This can be easily done using representation~(\ref{sigmakappavartheta}) for $\sigma_{\kappa,\vartheta}$. Because the explicit expressions for $\mathscr M_{\kappa,\vartheta}$ differ for $0<|\kappa|<1$ and $\kappa=0$, we consider these cases separately.

\par\medskip\noindent 1. The case $0<|\kappa|<1$: In view of~(\ref{Wuv}) and~(\ref{Mkappavartheta}), we have
\begin{equation}\label{eee}
\mathscr M_{\kappa,\vartheta}(z) = \frac{1}{2}\frac{\cos(\vartheta+\vartheta_\kappa) - e^{-i\pi\kappa}z^\kappa\cos(\vartheta-\vartheta_\kappa)}{\sin(\vartheta+\vartheta_\kappa) - e^{-i\pi\kappa}z^\kappa\sin(\vartheta-\vartheta_\kappa)}.
\end{equation}
It is easy to see that $\mathscr M_{\kappa,\vartheta}$ has no singularities on $(0,\infty)$ and
\[
\mathrm{Im}\,\mathscr M_{\kappa,\vartheta}(E) = \frac{1}{2}\frac{\sin^2\pi\kappa}{E^{-\kappa}\sin^2\vartheta_+ - 2\cos\pi\kappa\sin\vartheta_+\sin\vartheta_- + E^{\kappa}\sin^2\vartheta_-},\quad E>0,
\]
where $\vartheta_\pm = \vartheta\pm \vartheta_\kappa$. By~(\ref{measVkappatheta}), (\ref{tildeVkappatheta}), and~(\ref{sigmakappavartheta}), we conclude that $\sigma_{\kappa,\vartheta}$ coincides with $\mathcal V_{\kappa,\vartheta}$ on $(0,\infty)$. For $\vartheta \in [-|\vartheta_\kappa|,|\vartheta_\kappa|]+\pi\Z$, $\mathscr M_{\kappa,\vartheta}$ is real on $(-\infty,0)$ and has no singularities on this set. Formula~(\ref{sigmakappavartheta}) therefore implies that $\sigma_{\kappa,\vartheta}$ is zero on $(-\infty,0)$ for such $\vartheta$. If $\vartheta\in (|\vartheta_\kappa|,\pi-|\vartheta_\kappa|)+\pi\Z$, then $\mathscr M_{\kappa,\theta}$ has a simple pole at the point $E_{\kappa,\vartheta}$ given by~(\ref{Ekappavartheta}) and, hence, is representable in the form
\[
\mathscr M_{\kappa,\vartheta}(z) = g(z) + \frac{A}{E_{\kappa,\vartheta}-z},
\]
where $g$ is a function analytic in $\C_{3\pi/2}$ and real on $(-\infty,0)$ and
\[
A = \lim_{z\to E_{\kappa,\vartheta}} (E_{\kappa,\vartheta}-z)\mathscr M_{\kappa,\vartheta}(z) = \frac{\sin\pi\kappa|E_{\kappa,\vartheta}|}{2\kappa\sin(\vartheta+\vartheta_\kappa)\sin(\vartheta-\vartheta_\kappa)}.
\]
It therefore follows from~(\ref{sigmakappavartheta}) that $\sigma_{\kappa,\vartheta}$ is equal to $\pi A\delta_{E_{\kappa,\vartheta}}$ on $(-\infty,0)$. Hence,  $\sigma_{\kappa,\vartheta}$ coincides with $\mathcal V_{\kappa,\vartheta}$ on $(-\infty,0)$ for all $\vartheta$.

\par\medskip\noindent 2. The case $\kappa=0$: In view of~(\ref{Wuv0}) and~(\ref{Mkappavartheta}), we have
\[
\mathscr M_{0,\vartheta}(z) = \frac{1}{2}\,\frac{(i-\pi^{-1}\log z)\cos\vartheta - \sin\vartheta}{\cos\vartheta + (i-\pi^{-1}\log z)\sin\vartheta}.
\]
It is easy to see that $\mathscr M_{0,\vartheta}$ has no singularities on $(0,\infty)$ and
\[
\mathrm{Im}\,\mathscr M_{0,\vartheta}(E) = \frac{1}{2}\,\frac{1}{ (\cos\vartheta -\log E\sin\vartheta/\pi)^2 + \sin^2\vartheta},\quad E>0.
\]
By~(\ref{measV0theta}), (\ref{tildeV0theta}), and~(\ref{sigmakappavartheta}), we conclude that $\sigma_{0,\vartheta}$ coincides with $\mathcal V_{0,\vartheta}$ on $(0,\infty)$. For $\vartheta\in \pi\Z$, $\mathscr M_{0,\vartheta}$ is real on $(-\infty,0)$ and has no singularities on this set. Formula~(\ref{sigmakappavartheta}) therefore implies that $\sigma_{0,\vartheta}$ is zero on $(-\infty,0)$ for such $\vartheta$. If $\vartheta\notin \pi\Z$, then $\mathscr M_{0,\theta}$ has a simple pole at the point $E_{0,\vartheta}$ given by~(\ref{E0vartheta}) and is hence representable in the form
\[
\mathscr M_{0,\vartheta}(z) = g(z) + \frac{A}{E_{0,\vartheta}-z},
\]
where $g$ is a function analytic in $\C_{3\pi/2}$ and real on $(-\infty,0)$ and
\[
A = \lim_{z\to E_{0,\vartheta}} (E_{0,\vartheta}-z)\mathscr M_{0,\vartheta}(z) = \frac{\pi|E_{0,\vartheta}|}{2\sin^2\vartheta}.
\]
It therefore follows from~(\ref{sigmakappavartheta}) that $\sigma_{0,\vartheta}$ is equal to $\pi A\delta_{E_{0,\vartheta}}$ on $(-\infty,0)$. Therefore,  $\sigma_{0,\vartheta}$ coincides with $\mathcal V_{0,\vartheta}$ on $(-\infty,0)$ for all $\vartheta$.
\end{proof}

\begin{proof}[Proof of Theorem~$\ref{leig2}$]
It follows from statement~2 in Proposition~\ref{t_eig} and Lemma~\ref{l_spec1} that the operator $U_{\kappa,\vartheta}$ exists and is equal to the spectral transformation for $\mathfrak t_{\kappa,\vartheta}$. Statement~3 in Proposition~\ref{t_eig} and Lemma~\ref{l_spec1} therefore imply that $h_{\kappa,\vartheta}$ is equal to $L_q(X^{\mathfrak t_{\kappa,\vartheta}}\cap \mathcal D^r_{q_\kappa})$. By statement~1 in Proposition~\ref{p_bound} and formula~(\ref{hkappa}), we conclude that $h_{\kappa,\vartheta}$ is a self-adjoint extension of $h_\kappa$. In view of Lemma~\ref{l_triple} and~(\ref{Lqf}), we have
\begin{equation}\label{hkappavartheta}
h_{\kappa,\vartheta} = L_{q_\kappa}^{u^\kappa_\vartheta(0)}.
\end{equation}
By~(\ref{ukappavartheta}) and~(\ref{wruw}), every real $f\in \mathcal D$ satisfying $l_{q_\kappa}f=0$ is proportional to $u^\kappa_\vartheta(0)$ for some $\vartheta\in\R$. By Lemma~\ref{l_lqf} and formulas~(\ref{hkappa}) and~(\ref{hkappavartheta}), it follows that every self-adjoint extension of $h_\kappa$ is equal to $h_{\kappa,\vartheta}$ for some $\vartheta\in\R$. Let $\vartheta,\vartheta'\in\R$. By Lemma~\ref{l_lqf} and~(\ref{hkappavartheta}), we have $h_{\kappa,\vartheta}=h_{\kappa,\vartheta'}$ if and only if $u^\kappa_\vartheta(0) = c u^\kappa_{\vartheta'}(0)$ for some real $c\neq 0$. In view of~(\ref{ukappavartheta}) and~(\ref{wruw}), the last condition holds if and only if $\vartheta-\vartheta'\in \pi\Z$.
\end{proof}

\begin{remark}
We note that the function $q_\kappa$ given by~(\ref{qkappa}) is real not only for real $\kappa$ but also for imaginary $\kappa$. A complete description of eigenfunction expansions in this case can be found in~\cite{GTV2010}. It is easy to see that $\mathfrak t_{\kappa,\vartheta}=(q_\kappa,\mathcal D^r_{q_\kappa},u^\kappa_\vartheta)$ remains an expansion triple for imaginary $\kappa$ and the spectral measure for $\mathfrak t_{\kappa,\vartheta}$ can again be calculated using formulas~(\ref{sigmakappavartheta}) and~(\ref{eee}). An analogue of Theorem~\ref{leig2} for imaginary $\kappa$ can thus be obtained.
\end{remark}

\section{Continuity of spectral measures}
\label{s4}

In this section, we prove Theorem~\ref{l_borel}.

Let the continuous function $\Phi$ on $(-1,1)\times \R_+$ be defined by setting
\begin{equation}\label{gkappaE0}
\Phi(\kappa,E) = -\frac{\log E}{\pi \sinc (\pi\kappa)}\sinc\left(\frac{i\kappa}{2}\log E\right),
\end{equation}
where the entire function $\sinc$ is defined by the equality
\begin{equation}\nonumber
\sinc\zeta = \left\{
\begin{matrix}
\zeta^{-1}\sin\zeta,& \zeta\in\C\setminus\{0\},\\
1,& \zeta=0.
\end{matrix}
\right.
\end{equation}
It follows that
\begin{equation}\label{gkappaE}
\Phi(\kappa,E) = \left\{
\begin{matrix}
-\log E/\pi,& \kappa=0,\\
(\sin\pi\kappa)^{-1}(E^{-\kappa/2}-E^{\kappa/2}),& 0<|\kappa|<1.
\end{matrix}
\right.
\end{equation}
For every $\vartheta\in\R$ and $-1<\kappa<1$, we define the function $t_{\kappa,\vartheta}$ on $\R_+$ by the formula
\begin{equation}\label{tkappavartheta}
t_{\kappa,\vartheta}(E) = 2 + \Phi(\kappa,E)^2(1-\cos 2\vartheta\cos\pi\kappa) + \Phi(\kappa,E) (E^{-\kappa/2}+E^{\kappa/2}) \sin 2\vartheta.
\end{equation}
It follows from~(\ref{tildeVkappatheta}), (\ref{tildeV0theta}), and~(\ref{gkappaE}) by a straightforward calculation that
\begin{equation}\label{dtilde}
d\tilde{\mathcal V}_{\kappa,\vartheta}(E) = t_{\kappa,\vartheta}(E)^{-1}\Theta(E)\,dE
\end{equation}
for all $\vartheta\in\R$ and $-1<\kappa<1$. By the Cauchy--Bunyakovsky inequality, we have
\[
|-c\cos 2\vartheta + d\sin 2\vartheta| \leq \sqrt{c^2+d^2}
\]
for any $c,d\in\R$. Applying this bound to $c = \Phi(\kappa,E)^2\cos\pi\kappa$ and
\[
d = \Phi(\kappa,E) (E^{-\kappa/2}+E^{\kappa/2}) = \Phi(\kappa,E)\sqrt{\Phi(\kappa,E)^2\sin^2\pi\kappa+4},
\]
from~(\ref{tkappavartheta}), we deduce that $t_{\kappa,\vartheta}(E)\geq f(\Phi(\kappa,E)^2)$, where $f(y) = 2+y-\sqrt{y^2+4y}$, $y\geq 0$. Because
\[
f(y) = \frac{4}{2+y+\sqrt{y^2+4y}}\geq \frac{2}{2+y},\quad y\geq 0,
\]
we conclude that $t_{\kappa,\vartheta}(E)^{-1}\leq 1 + \Phi(\kappa,E)^2/2$ for all $E>0$, $-1<\kappa<1$, and $\vartheta\in \R$. By~(\ref{gkappaE0}), the function $\kappa\to \Phi(\kappa,E)^2$ is even and increases on $[0,1)$ for every $E>0$. Let $0<\alpha<1$. In view of~(\ref{gkappaE}), it follows that
\begin{equation}\label{boundbound}
t_{\kappa,\vartheta}(E)^{-1}\leq 1 + \frac{1}{2}\Phi(\alpha,E)^2 \leq \frac{1}{2\sin^2\pi\alpha}(E^{\alpha}+E^{-\alpha})
\end{equation}
for all $E>0$, $\vartheta\in\R$, and $-\alpha\leq\kappa\leq\alpha$. Let $\varphi$ be a bounded Borel function on $\R$ with compact support and $B=(-1,1)\times\R$. Because the function $(\kappa,\vartheta)\to t_{\kappa,\vartheta}(E)^{-1}\varphi(E)$ is continuous on $B$ for every $E>0$, relations~(\ref{dtilde}) and~(\ref{boundbound}) and the dominated convergence theorem imply that $(\kappa,\vartheta)\to \int \varphi(E)\,d\tilde{\mathcal V}_{\kappa,\vartheta}(E)$ is a continuous function on $B$ that is bounded on $[-\alpha,\alpha]\times \R$ for every $0\leq \alpha<1$. Let
\[
B' = \{(\kappa,\vartheta)\in B: \vartheta\in (|\vartheta_\kappa|,\pi-|\vartheta_\kappa|)+\pi\Z\}.
\]
It follows from~(\ref{measVkappatheta}) and~(\ref{measV0theta}) that
\[
\int \varphi(E)\,d{\mathcal V}_{\kappa,\vartheta}(E) =\int \varphi(E)\,d\tilde{\mathcal V}_{\kappa,\vartheta}(E) + b_\varphi(\kappa,\vartheta),
\]
where the function $b_\varphi$ on $B$ is defined by the relation
\begin{equation}\label{b_F}
b_\varphi(\kappa,\vartheta) = \left\{
\begin{matrix}
\tilde \Phi(\kappa,|E_{\kappa,\vartheta}|)\varphi(E_{\kappa,\vartheta}),& (\kappa,\vartheta)\in B',\\
0,& (\kappa,\vartheta)\in B\setminus B',
\end{matrix}
\right.
\end{equation}
and the continuous function $\tilde \Phi$ on $(-1,1)\times \R_+$ is given by
\[
\tilde \Phi(\kappa,E) = \frac{1}{2} E\pi^2 \sinc(\pi\kappa) \left(\Phi(\kappa,E)^2+\frac{1}{\cos^2\vartheta_\kappa}\right).
\]
For every $(\kappa,\vartheta)\in B'$, we have $|\cot\vartheta\tan\vartheta_\kappa|<1$, and it follows from~(\ref{Ekappavartheta}) and (\ref{E0vartheta}) that
\[
E_{\kappa,\vartheta} = -\exp\left[\frac{\pi\cot\vartheta}{2\cos\vartheta_\kappa}\sinc(\vartheta_\kappa)g(\cot\vartheta\tan\vartheta_\kappa) \right],
\]
where $g$ is a continuous function on $(-1,1)$ such that $g(y)=y^{-1}\log((1+y)(1-y)^{-1})$ for $y\neq 0$ and $g(0)=2$. Hence, $(\kappa,\vartheta)\to E_{\kappa,\vartheta}$ is a continuous function on $B'$, and $b_\varphi$ is therefore a Borel function on $B$. Estimating $\Phi(\kappa,E)^2$ as above, we obtain
\begin{equation}\label{estest}
\tilde \Phi(\kappa,E) \leq \frac{\pi^2E}{2\sin^2\pi\alpha}(E^\alpha + E^{-\alpha}),\quad (\kappa,E)\in [-\alpha,\alpha]\times\R_+,
\end{equation}
for every $0<\alpha<1$. In view of~(\ref{b_F}), this implies that $b_\varphi$ is bounded on $[-\alpha,\alpha]\times\R$ for every $0\leq\alpha<1$. To complete the proof, it remains to show that $b_\varphi$ is continuous on $B$ if $\varphi$ is continuous. Let $-1<\kappa<1$. It follows from~(\ref{Ekappavartheta}) and (\ref{E0vartheta}) that $|E_{\kappa,\vartheta}|$ strictly decreases from $\infty$ to $0$ as $\vartheta$ varies from $|\vartheta_\kappa|$ to $\pi-|\vartheta_\kappa|$. Hence, for every $E>0$, there is a unique $\tau_E(\kappa)\in (|\vartheta_\kappa|,\pi-|\vartheta_\kappa|)$ such that $|E_{\kappa,\tau_E(\kappa)}|=E$. The continuity of $E_{\kappa,\vartheta}$ in $(\kappa,\vartheta)$ implies that $\tau_E$ is a continuous function on $(-1,1)$ for every $E>0$. Let $\beta>0$ be such that $\varphi(E)=0$ for every $E\leq -\beta$. Given $0<\alpha<1$ and $0<\delta$, we define the open subset $B_{\alpha,\delta}$ of $B$ by setting
\[
B_{\alpha,\delta} = \{(\kappa,\vartheta)\in (-\alpha,\alpha)\times \R : \vartheta\in (\tau_\delta(\kappa)-\pi, \tau_\beta(\kappa))+\pi\Z\}.
\]
If $\delta<1$, then it follows from~(\ref{b_F}) and~(\ref{estest}) that
\[
|b_\varphi(\kappa,\vartheta)| \leq \frac{\pi^2\delta^{1-\alpha}}{\sin^2\pi\alpha}\sup_{E\in\R} |\varphi(E)|,\quad (\kappa,\vartheta)\in B_{\alpha,\delta}.
\]
Given $(\kappa,\vartheta)\in B\setminus B'$ and $\varepsilon>0$, we pick an arbitrary $\alpha\in (|\kappa|,1)$ and choose $\delta>0$ so small that the right-hand side of the last inequality is less than $\varepsilon$. Then $B_{\alpha,\delta}$ is a neighborhood of $(\kappa,\vartheta)$ where the absolute value of $b_\varphi$ is less than $\varepsilon$. This proves that $b_\varphi$ is continuous at every point of $B\setminus B'$. Because $b_\varphi$ is obviously continuous on $B'$, the theorem is proved.

\appendix
\section{Proof of Lemma~\ref{l_analyt}}
\label{appA}
Let $\mathrm{Log}$ be the branch of the logarithm in $\C_\pi$ satisfying $\mathrm{Log}\,1 = 0$ and $p$ be the analytic function in $\C\times \C_\pi$ defined by the relation $p(\kappa,r) = e^{\kappa\,\mathrm{Log}\,r}$ (hence $p(\kappa,r) = r^\kappa$ for $r\in\R_+$). Let $G$ be the analytic function in $\C\times\C\times \C_\pi$ such that $G(\kappa,z,r) = p(1/2+\kappa,r)\mathcal X_\kappa(r^2z)$ for all $\kappa,z\in\C$ and $r\in \C_\pi$. We then have $G(\kappa,z,r) = u^\kappa(z|r)$ for all $\kappa,z\in\C$ and $r\in\R_+$. We define the function $F$ on $\mathscr O\times\C\times\C\times \C_\pi$ by setting
\begin{align}
& F(\kappa,\vartheta,z,r) = \frac{G(\kappa,z,r)\sin(\vartheta+\vartheta_\kappa) - G(-\kappa,z,r)\sin(\vartheta-\vartheta_\kappa)}{\sin\pi\kappa},\quad\kappa\in \mathscr O\setminus\{0\},\nonumber\\
& F(0,\vartheta,z,r) = G(0,z,r)\cos\vartheta+\frac{2}{\pi}\left[\left(\mathrm{Log}\,\frac{r}{2} + \gamma\right)G(0,z,r) - p(1/2,r)\,\mathcal Y(z r^2)\right]\sin\vartheta \nonumber
\end{align}
for every $z,\vartheta\in\C$ and $r\in \C_\pi$. It follows immediately from~(\ref{wkappa}), (\ref{w(z)}), and the definition of $F$ that $F(\kappa,\vartheta,z,r) = u^\kappa_\vartheta(z|r)$ for every $\vartheta,z\in\C$, $\kappa\in\mathscr O$, and $r\in\R_+$. The function $(\vartheta,z,r)\to F(\kappa,\vartheta,z,r)$ is obviously analytic in $\C\times\C\times \C_\pi$ for every fixed $\kappa\in\mathscr O$. The function $\kappa\to F(\kappa,\vartheta,z,r)$ is analytic in $\mathscr O\setminus \{0\}$ and continuous at $\kappa = 0$ (this is ensured by the same calculation as used to find the limit in~(\ref{w(z)})) and is therefore analytic in $\mathscr O$ for every fixed $\vartheta,z\in\C$ and $r\in \C_\pi$. Hence, $F$ is analytic in $\mathscr O\times\C\times\C\times \C_\pi$ by the Hartogs theorem.

\section{Proof of Lemma~\ref{ll}}
\label{app}
Let $T=L_q(C^\infty_0(a,b))$. Because $L_q$ is closed, it suffices to show that $T^* = L_q^*$. For this, we only need to prove that $D_{T^*}\subset D_{L_q^*}$ because $L_q$ is an extension of $T$ and $T^*$ is hence an extension of $L_q^*$. By~(\ref{lq*}), $D_{L_q^*}$ consists of all elements $[g]$ with $g\in \mathcal D_q$. Therefore, for every $\phi\in D_{T^*}$, we must find $g\in \mathcal D_q$ such that $\phi = [g]$.
Let $\psi = T^*\phi$. We then have
\[
\langle T[f],\phi\rangle = \langle [f],\psi\rangle,\quad f\in C_0^\infty(a,b).
\]
Because $(T[f])(r) = -f''(r)+q(r)f(r)$ for $\lambda$-a.e. $r\in(a,b)$, we obtain
\[
-\int_a^b \overline{f''(r)}\phi(r)\,dr = \int_a^b \overline{f(r)}(\psi(r)-q(r)\phi(r))\,dr,\quad f\in C_0^\infty(a,b).
\]
Because both $q$ and $\phi$ are locally square-integrable on $(a,b)$, the function $q\phi$ is locally integrable on $(a,b)$. We choose $c\in (a,b)$ and define $h\in \mathcal D$ by setting
\[
h(r) = \int_c^r d\rho\int_c^{\rho}(\psi(t)-q(t)\phi(t))\,dt.
\]
We obviously have $h''(r) = \psi(r)-q(r)\phi(r)$ for $\lambda$-a.e. $r\in (a,b)$. Integrating by parts, we obtain
\[
\int_a^b \overline{f(r)}(\psi(r)-q(r)\phi(r))\,dr = \int_a^b \overline{f''(r)}h(r)\,dr
\]
and therefore
\[
\int_a^b \overline{f''(r)}(\phi(r)+h(r))\,dr = 0,\quad f\in C_0^\infty(a,b).
\]
This means that the second derivative of $\phi + h$ in the sense of generalized functions is equal to zero. Hence, there are $A,B\in\C$ such that $\phi(r) + h(r) = Ar+B$ for $\lambda$-a.e. $r\in (a,b)$. Let $g\in \mathcal D$ be defined by the relation $g(r) = Ar+B-h(r)$, $r\in (a,b)$. Then we obviously have $[g] = \phi$. Because $g''(r) = -\psi(r) + q(r)\phi(r)$, we have $(l_q g)(r) = \psi(r) - q(r)\phi(r) + q(r)g(r) = \psi(r)$ for $\lambda$-a.e. $r\in(a,b)$ and therefore $l_qg = \psi$. This implies that both $g$ and $l_qg$ are square-integrable and hence $g\in \mathcal D_q$.

\section*{Acknowledgments}
The author is grateful to I.V.~Tyutin and B.L.~Voronov for the useful discussions.

\bibliographystyle{unsrt}
\bibliography{isquare}

\begin{thebibliography}{10}

\bibitem{Hankel}
Herrmann Hankel.
\newblock Die {F}ourier'schen {R}eihen und {I}ntegrale f\"ur
  {C}ylinderfunctionen.
\newblock {\em Math. Ann.}, 8(4):471--494, 1875.

\bibitem{Weyl}
Hermann Weyl.
\newblock \"{U}ber gew\"ohnliche {D}ifferentialgleichungen mit
  {S}ingularit\"aten und die zugeh\"origen {E}ntwicklungen willk\"urlicher
  {F}unktionen.
\newblock {\em Math. Ann.}, 68(2):220--269, 1910.

\bibitem{Weyl1}
Hermann Weyl.
\newblock {\"U}ber gew\"ohnliche lineare {D}ifferentialgleichungen mit
  singul\"aren {S}tellen und ihre {E}igenfunktionen (2. note).
\newblock {\em G{\"o}tt. Nachr.}, pages 442--467, 1910.

\bibitem{Titchmarsh}
E.~C. Titchmarsh.
\newblock {\em Eigenfunction expansions associated with second-order
  differential equations. {P}art {I}}.
\newblock Second Edition. Clarendon Press, Oxford, 1962.

\bibitem{Naimark}
M.~A. Na{\u\i}mark.
\newblock {\em Linear differential operators. {P}art {II}: {L}inear
  differential operators in {H}ilbert space}.
\newblock Frederick Ungar Publishing Co., New York, 1968.

\bibitem{GesztesyZinchenko}
Fritz Gesztesy and Maxim Zinchenko.
\newblock On spectral theory for {S}chr\"odinger operators with strongly
  singular potentials.
\newblock {\em Math. Nachr.}, 279(9-10):1041--1082, 2006.

\bibitem{Fulton}
Charles Fulton.
\newblock Titchmarsh-{W}eyl {$m$}-functions for second-order
  {S}turm-{L}iouville problems with two singular endpoints.
\newblock {\em Math. Nachr.}, 281(10):1418--1475, 2008.

\bibitem{KST}
Aleksey Kostenko, Alexander Sakhnovich, and Gerald Teschl.
\newblock Weyl-{T}itchmarsh theory for {S}chr\"odinger operators with strongly
  singular potentials.
\newblock {\em Int. Math. Res. Not.}, 2012:1699--1747, 2012.

\bibitem{GTV2010}
D.~M. Gitman, I.~V. Tyutin, and B.~L. Voronov.
\newblock Self-adjoint extensions and spectral analysis in the {C}alogero
  problem.
\newblock {\em J. Phys. A}, 43(14):145205 (34pp), 2010.

\bibitem{Smirnov2015}
A.~G. Smirnov.
\newblock Reduction by symmetries in singular quantum-mechanical problems:
  {G}eneral scheme and application to {A}haronov-{B}ohm model.
\newblock {\em J. Math. Phys.}, 56:122101, 2015.

\bibitem{EverittKalf}
W.~Norrie Everitt and Hubert Kalf.
\newblock The {B}essel differential equation and the {H}ankel transform.
\newblock {\em J. Comput. Appl. Math.}, 208(1):3--19, 2007.

\bibitem{Bateman}
Arthur Erd{\'e}lyi, Wilhelm Magnus, Fritz Oberhettinger, and Francesco~G.
  Tricomi.
\newblock {\em Higher transcendental functions. {V}ols. {I}, {II}}.
\newblock McGraw-Hill Book Company, Inc., New York-Toronto-London, 1953.
\newblock Based, in part, on notes left by Harry Bateman.

\bibitem{Kodaira}
Kunihiko Kodaira.
\newblock The eigenvalue problem for ordinary differential equations of the
  second order and {H}eisenberg's theory of {$S$}-matrices.
\newblock {\em Amer. J. Math.}, 71:921--945, 1949.

\bibitem{Teschl}
Gerald Teschl.
\newblock {\em Mathematical methods in quantum mechanics}, volume~99 of {\em
  Graduate Studies in Mathematics}.
\newblock American Mathematical Society, Providence, RI, 2009.

\bibitem{Weidmann}
Joachim Weidmann.
\newblock {\em Spectral theory of ordinary differential operators}, volume 1258
  of {\em Lecture Notes in Mathematics}.
\newblock Springer-Verlag, Berlin, 1987.

\bibitem{BennewitzEveritt}
Christer Bennewitz and W.~Norrie Everitt.
\newblock The {T}itchmarsh-{W}eyl eigenfunction expansion theorem for
  {S}turm-{L}iouville differential equations.
\newblock In {\em Sturm-{L}iouville theory}, pages 137--171. Birkh\"auser,
  Basel, 2005.

\end{thebibliography}
\end{document}